\documentclass[journal]{IEEEtran}

\usepackage{graphicx,booktabs,amsmath,amssymb,microtype,url}
\usepackage{algorithm}
\usepackage{algpseudocode}
\usepackage{tikz}
\usetikzlibrary{arrows.meta,positioning,fit,backgrounds,calc}
\graphicspath{{figures/}}
\usepackage[hidelinks]{hyperref}

\newtheorem{theorem}{Theorem}
\newtheorem{corollary}{Corollary}
\newtheorem{definition}{Definition}
\newtheorem{assumption}{Assumption}

\begin{document}

\title{Adaptive Participation Under Statically Equivalent\\
Incentives in Distributed Demand Response Systems}

\author{Xun~Shao,~\IEEEmembership{Senior~Member,~IEEE,}
        Ryoichi~Inoue, Shinken~Takekawa, and~Go~Hasegawa,~\IEEEmembership{Member,~IEEE}%
\thanks{X. Shao, R. Inoue and S. Takekawa are with the Department of Electrical and
Electronic Information Engineering, Toyohashi University of Technology, Toyohashi, Japan
(e-mail: shao.xun.ls@tut.jp). \emph{(Corresponding author: Xun Shao.)}}%
\thanks{G. Hasegawa is with the Research Institute of Electrical Communication, Tohoku
University, Sendai, Japan.}%
\thanks{This work has been submitted to the IEEE for possible publication. Copyright may be
transferred without notice, after which this version may no longer be accessible.}}

\maketitle

\begin{abstract}
Aggregators recruit distributed energy resources with settlement rules and participation
payments. Such designs are normally validated at fixed points: zero participation must cease to
be an equilibrium, and truthful capability reporting must remain a best reply. We ask whether
those checks determine the participation that owners reach once they adapt from the settlements
they receive. In a five-unit event with fixed dispatch, payment rule and penalty, we vary only
how a scarcity-contingent participation payment decays with the capability others have declared.
Of two decay structures that agree on all five static criteria, one reaches full participation
from a collapse initialization in 96 of 96 seeds and the other in none, within an 8000-round
horizon and with disjoint 95\% confidence intervals. The difference lies in the payoffs offered at partial
participation, which the static criteria never evaluate; it is a property of experience-based
feedback and closes when counterfactual payoffs are supplied. Because those payoffs make each
unit's settlement depend on what the others declared, we also ask what survives when the
mechanism is distributed. Running the aggregator and the five units as separate processes
reproduced the centralized reference at every round, and a deliberately misattributed
declaration was detected although every message was delivered.
\end{abstract}


% ============================================================ I
\section{Introduction}
\label{sec:intro}

\IEEEPARstart{A}{ggregators} recruit distributed energy resources into demand-response programs
by offering settlement rules and participation payments. Designing those payments is an active
area: mechanisms have been proposed for truthful and privacy-aware demand
response~\cite{tsaousoglou20,muthirayan20}, for flexibility procurement across
timescales~\cite{ikuta25}, for retail pricing~\cite{zheng26}, and for market clearing in
integrated and community settings~\cite{ge26,alahmed25}. A second line assumes that owners are
not fixed but learn: reinforcement learning has been applied to demand response with multiple
aggregators~\cite{fraija24}, safe pricing has been designed under bandit
feedback~\cite{hutchinson24}, and incentives have been adapted online against players who
update their own strategies~\cite{maheshwari26}.

These two lines meet at a validation question. A participation payment is normally accepted once
zero participation is no longer an equilibrium and truthful capability reporting remains a best
reply~\cite{vickrey61,myerson81}. Both conditions are evaluated at full and at zero
participation, and both describe a fixed point. Owners who learn from the settlements they
receive do not solve for that fixed point; they move through intermediate states on the way to
it.

We ask what those fixed-point checks leave undetermined. In a five-unit event with fixed
dispatch, payment rule, penalty and reporting calibration, we vary only the shape of a
scarcity-contingent participation payment: how it decays with the capability the other units have
declared. Two decay structures satisfy all five static criteria with identical entry thresholds,
so no standard check separates them. Under owners who observe only their own realized
settlement, one reaches full participation from a collapse initialization in 96 of 96 seeds and
the other in none, within the canonical 8000-round horizon and with disjoint 95\% confidence
intervals. The two structures pay the same amount to a lone entrant and nothing once the
capability target is met; they differ only at the intermediate states, and it is there that the
outcome is decided. The effect belongs to experience-based feedback: supplying counterfactual
payoffs removes it.

The transfer that produces this behavior is jointly coupled. What a unit is paid depends on what
the others declared, so the payoff a learner observes is shaped by the rest of the population.
That coupling also complicates deployment. A real program runs on physically separate devices:
each unit holds its own estimate of what its declarations are worth, and the aggregator holds the
settlement table. Decision, profile assembly, settlement and learning update then happen on
different machines, and a discrepancy in any of them can reach several units in the same round
and is written into state that persists. Counting delivered messages does not detect this, and
neither does comparing final outcomes, because different trajectories can end at the same
profile. For distributed coordination mechanisms whose devices carry adaptive state, what needs
evidence is the agreement of that state, round by round.

We therefore state a fidelity criterion over the quantities the mechanism is defined on and
evaluate a realization in which the aggregator and the five units run as separate processes over
content-centric forwarders~\cite{cefore}. The substrate is an implementation choice, reported so
the run can be reproduced; we propose no networking protocol and claim no advantage over other
messaging substrates.

\smallskip\noindent\textbf{Contributions.}
\begin{itemize}
\item A jointly coupled participation transfer for aggregator-coordinated demand response,
together with its static entry thresholds and its payoffs at partial participation, which the
static criteria do not evaluate.
\item Evidence that two transfers indistinguishable under those criteria lead adaptive owners to
different outcomes within the canonical horizon, and that the difference is located in the
partial-participation payoffs.
\item A per-round fidelity criterion for this class of mechanism and its evaluation on a
distributed realization, which reproduced the centralized reference exactly over the full
8000-round trajectory, together with a preregistered fault that the criterion detects despite
successful delivery.
\end{itemize}

The evidence is bounded to five units, one operating point and finite horizons.

% ============================================================ II
\section{Related Work}
\label{sec:related}

\subsection{Participation incentives in demand response, and their validation}

Aggregator-side demand-response designs are normally validated on equilibrium properties: a
payment rule is accepted once zero participation is no longer an equilibrium and truthful
reporting remains a best reply. This is the practice in aggregator mechanism
baselines~\cite{tsaousoglou20,muthirayan20}, in flexibility-procurement designs across
timescales~\cite{ikuta25}, in consumer-facing retail pricing~\cite{zheng26,kim24}, and in integrated market clearing and
community-market settings~\cite{ge26,alahmed25,xu25}; the underlying framework is
truthful auction and optimal mechanism design~\cite{vickrey61,myerson81}, with zero
participation as a coordination failure in the global-games sense~\cite{morrisshin}. The
validation criteria relevant to the present comparison are evaluated at fixed points.

A second line assumes participants adapt. Reinforcement learning has been applied to
demand response with multiple aggregators~\cite{fraija24}, and safe pricing under bandit
feedback has been designed for distributed resource allocation~\cite{hutchinson24}. Closest to
the present setting, adaptive incentive design with learning agents establishes a two-timescale
system in which the mechanism updates more slowly than the players and converges to a fixed
point that is socially optimal~\cite{maheshwari26}.

That work designs a mechanism for adaptive participants and asks whether it converges to a
desirable fixed point. The present study asks the converse question about the validation
practice itself: whether two transfers that a static validation cannot distinguish are
distinguished by the adaptive outcome. Constructing a static equivalence class over named
criteria and then measuring whether the learning dynamics respect it is the approach taken here;
we did not find it treated in the demand-response or adaptive-incentive literature surveyed for
this section. The learning rule itself is adopted unchanged from the learning-in-games
literature~\cite{hartmascolell00}, with fictitious play~\cite{robinson51}, stochastic
stability~\cite{kandori93} and independent reinforcement learners~\cite{tan93} as the standard
alternatives.

\subsection{Realizing aggregator--DER coordination as a distributed system}

Aggregator--DER coordination has been realized as a distributed system in several forms:
decentralized demand response across coupled energy carriers~\cite{wang24}, and distributed
dispatch evaluated under persistent packet loss~\cite{ren24}. Communication-architecture studies,
including information-centric proposals for smart-grid communications~\cite{ameme17}, evaluate
latency, loss, retransmission and cache behavior. These evaluate the transport or the delivered
service.
The present study takes the substrate as given---it is instantiated on the Cefore CCNx
platform~\cite{cefore} and its emulator, without any claim of advantage over other messaging
substrates---and asks a question the transport evaluations do not: whether the mechanism's own
quantities survive the distribution.

\subsection{Establishing that a distributed implementation realizes its reference}

Checking an implementation against a reference is established practice. Conformance testing uses
an executable model as both test generator and oracle, with a conformance relation based on
trace inclusion: traces of events are recorded in the system under test and replayed in the
model~\cite{aichernig08}. Decentralized runtime verification monitors temporal-logic properties
over a running distributed system~\cite{mostafa15}.

Each of these establishes something our setting needs, and none establishes what it requires.
Conformance relations are established over an observation interface, whereas the quantities
whose agreement determines whether this mechanism was executed are private by construction: a
participant's estimate and visit count never leave it, and the assembled joint profile is never
published. Over the observables this realization exposes, an implementation that applies a
correct payoff to the wrong internal entry emits the same traces a faithful one emits, so
agreement on those traces does not settle the question here; a relation defined over a richer
observation interface would see more. Property monitoring is likewise insufficient on its own,
because the property at issue is equality of the whole state trajectory with a reference, not
satisfaction of a specified temporal formula.

We therefore state a semantic-fidelity criterion for this class of mechanism---jointly coupled
settlement, private per-agent learning state, feedback from settlement into learning---and
require per-round equality of every quantity in the dependency chain. This is a conformance-style
relation extended to state that the realization does not expose, specialized to a mechanism
class; it is not a new verification technique, and it is not offered as a general theory.

% ============================================================ III
\section{Mechanism}
\label{sec:model}

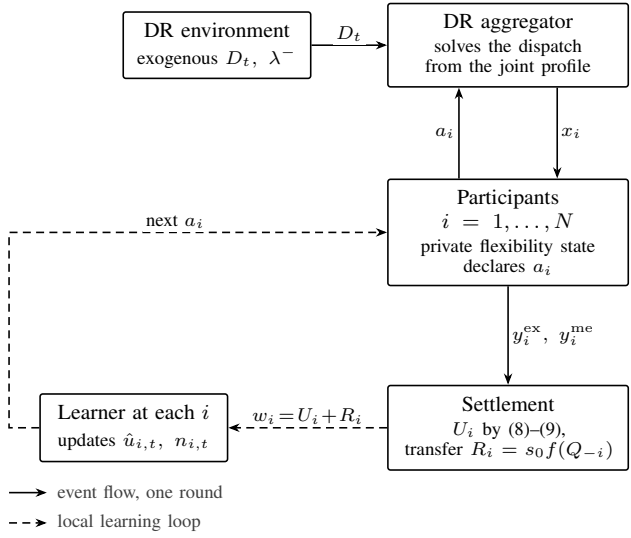
\begin{figure}[t]
\centering
% Fig. 1 — system schematic, drawn from TikZ primitives (nodes, paths, labels).
% No raster, no traced artwork, no external image asset.
% Symbols follow Section III: a_i declaration, x_i commanded discharge,
% y^ex/y^me executed and metered delivery, w_i = U_i + R_i realised settlement.
\begin{tikzpicture}[
  font=\footnotesize,
  >={Stealth[length=1.5mm]},
  box/.style   ={draw, semithick, rounded corners=1pt, align=center, inner sep=1.4mm,
                 minimum height=8mm, text width=22mm},
  wide/.style  ={box, text width=29mm},
  flow/.style  ={->, semithick},
  learn/.style ={->, semithick, densely dashed},
  lab/.style   ={font=\scriptsize, inner sep=1.2pt, fill=white},
  note/.style  ={font=\scriptsize, text=black!75, inner sep=1pt}]

% ---- nodes -------------------------------------------------------------
\node[box]                        (env) {DR environment\\[1pt]
                                         \scriptsize exogenous $D_t,\ \lambda^-$};
\node[wide, right=10mm of env]    (agg) {DR aggregator\\[1pt]
                                         \scriptsize solves the dispatch\\[-1.5pt]
                                         \scriptsize from the joint profile};
\node[wide, below=12mm of agg]    (par) {Participants $i=1,\dots,N$\\[1pt]
                                         \scriptsize private flexibility state\\[-1.5pt]
                                         \scriptsize declares $a_i$};
\node[wide, below=13mm of par]    (set) {Settlement\\[1pt]
                                         \scriptsize $U_i$ by
                                         \eqref{eq:payment}--\eqref{eq:utility},\\[-1.5pt]
                                         \scriptsize transfer $R_i=s_0f(Q_{-i})$};
\node[box, left=21mm of set]      (lrn) {Learner at each $i$\\[1pt]
                                         \scriptsize updates $\hat u_{i,t},\ n_{i,t}$};

% ---- event flow --------------------------------------------------------
\draw[flow] (env) -- node[lab, above] {$D_t$} (agg);
\draw[flow] ($(agg.south)+(6.5mm,0)$) -- node[lab, right, pos=.5] {$x_i$}
            ($(par.north)+(6.5mm,0)$);
\draw[flow] ($(par.north)-(6.5mm,0)$) -- node[lab, left, pos=.5] {$a_i$}
            ($(agg.south)-(6.5mm,0)$);
\draw[flow] (par) -- node[lab, right, pos=.5] {$y^{\mathrm{ex}}_i,\ y^{\mathrm{me}}_i$} (set);

% ---- learning loop -----------------------------------------------------
\draw[learn] (set.west) -- node[lab, above, pos=.5] {$w_i\!=\!U_i\!+\!R_i$} (lrn.east);
\draw[learn] (lrn.west) -- ++(-4mm,0) |- node[lab, above, pos=.72] {next $a_i$} (par.west);

% ---- legend ------------------------------------------------------------
\begin{scope}[shift={($(lrn.south west)+(-4mm,-4mm)$)}]
  \draw[flow]  (0,0) -- (5mm,0);   \node[note, anchor=west] at (6mm,0)    {event flow, one round};
  \draw[learn] (0,-4mm) -- (5mm,-4mm);
  \node[note, anchor=west] at (6mm,-4mm) {local learning loop};
\end{scope}
\end{tikzpicture}
\caption{Aggregator coordinating $N$ distributed units over one demand-response event. The
figure depicts the \emph{economic} coordination---declaration, dispatch and settlement---and not
a communication architecture; the distributed realization and its message flows are the subject
of Fig.~\ref{fig:arch}.}
\label{fig:overview}
\end{figure}

\subsection{Event, units, and declarations}

Fig.~\ref{fig:overview} shows the setting. $N=5$ units are coordinated over a $24$-hour day at
hourly resolution. Each unit $i$ holds a
battery with usable capacity $13.5$~kWh, charge and discharge limits $5$~kW, and one-directional
efficiency $0.95$. Each unit has a \emph{private flexibility state} $u_i\in\{\text{normal},
\text{stressed}\}$, unobserved by the aggregator, whose only physical effect is on the usable
discharge power,
\begin{equation}
\bar{p}_i = P_{\max}\cdot
\begin{cases}
1, & u_i = \text{normal},\\
\sigma_{\mathrm{str}}, & u_i = \text{stressed},
\end{cases}
\label{eq:cap}
\end{equation}
with $P_{\max}=5$~kW and $\sigma_{\mathrm{str}}=0.5$. The state is drawn uniformly at the start
of the day and evolves as a symmetric two-state Markov chain with
$\Pr[u_{i,t+1}=u_{i,t}]=\rho=0.95$. The demand-response event is exogenous and deterministic: a
reduction of $D_t=P_D=15.0$~kW is requested in each of $K=2$ contiguous hours of highest total
price.

Before the event each unit either abstains or declares one of two contract items, conservative
or aggressive, with declared discharge limits $q_C=2.50$~kW and $q_A=3.00$~kW and availability
payments $b_C=0.584896$ and $b_A=0.677219$~\$. Write $q^{\mathrm{decl}}_i$ for the declared limit
of a participating unit and $q^{\mathrm{decl}}_i=0$ for a unit that abstains; truthful
declaration is conservative when stressed and aggressive when normal. Write
$\mathbf{a}=(a_1,\dots,a_N)$ for the \emph{joint declaration profile}. This object, not any
individual declaration, is what the aggregator acts on.

\subsection{Dispatch and settlement}

The aggregator solves its dispatch program \emph{from the joint profile} and obtains each unit's
commanded discharge power $x_i$. The command splits into a guaranteed and an above-guarantee
block,
\begin{equation}
g_i = \min\big(x_i,\ q^{\mathrm{decl}}_i,\ c^{\mathrm{CE}}_i\big),
\qquad z_i = x_i - g_i,
\label{eq:blocks}
\end{equation}
where $c^{\mathrm{CE}}_i$ is the aggregator's \emph{certainty-equivalent capability estimate}:
the expectation of \eqref{eq:cap} under its current belief $b_i=\Pr[u_i=\text{stressed}]$,
\begin{equation}
c^{\mathrm{CE}}_i = P_{\max}\big[(1-b_i) + \sigma_{\mathrm{str}}\,b_i\big],
\label{eq:cce}
\end{equation}
with $b_i$ maintained by a Bayes filter over the commanded and metered discharge and propagated
between hours at the persistence $\rho$ of Section~\ref{sec:model}-A. Execution is
limited by true capability and available energy, and the meter is noisy:
\begin{equation}
y^{\mathrm{ex}}_i = \min\big(x_i,\ \bar{p}_i,\ r_i\big),
\qquad
y^{\mathrm{me}}_i = \max\big(0,\ y^{\mathrm{ex}}_i + \varepsilon_i\big),
\label{eq:meas}
\end{equation}
with $r_i$ the energy headroom and $\varepsilon_i\sim\mathcal{N}(0,\varsigma^2)$,
$\varsigma=0.10$~kW. Three settlement quantities follow, accumulated over the $K$ event hours:
\begin{align}
D^g_i &= \textstyle\sum_{t}\min\big(y^{\mathrm{ex}}_{i,t},\,g_{i,t}\big)\,\Delta t,
  \label{eq:Dg}\\
D^z_i &= \textstyle\sum_{t}\big(y^{\mathrm{ex}}_{i,t}-\min(y^{\mathrm{ex}}_{i,t},g_{i,t})\big)\,\Delta t,
  \label{eq:Dz}\\
s_i   &= \textstyle\sum_{t}\max\big(g_{i,t}-y^{\mathrm{me}}_{i,t}-\epsilon,\,0\big)\,\Delta t,
  \label{eq:short-i}
\end{align}
where $s_i$ is the verified shortfall, measured against the guaranteed block alone with a
tolerance $\epsilon=3\varsigma=0.30$~kW. The settlement pays each participating unit
\begin{equation}
P_i = b_i - \Pi\,s_i + \kappa\,D^z_i ,
\label{eq:payment}
\end{equation}
with $\Pi=1.3925$~\$/kWh the penalty rate on verified shortfall and $\kappa=0.1176$~\$/kWh the
rate at which the above-guarantee block is paid. Delivered energy costs the unit $\kappa$ per kWh
in degradation, so its utility is
\begin{equation}
U_i = P_i - \kappa\big(D^g_i + D^z_i\big) = b_i - \Pi\,s_i - \kappa\,D^g_i .
\label{eq:utility}
\end{equation}
The above-guarantee block cancels: it is paid at exactly the rate it costs. A unit that abstains
receives $P_i=U_i=0$. At the operating rate $\Pi=1.3925$ truthful declaration is a best reply for
every unit; nothing in what follows re-tunes any constant fixed here.

\subsection{The corrective participation transfer}

Truthful reporting being a best reply does not make participation one. Under
\eqref{eq:payment}--\eqref{eq:utility}, universal abstention is a strict Nash equilibrium, for an
accounting reason rather than a punitive one: a single participant carries more of the requested
reduction than the availability payment was sized for, and is charged $\kappa$ volumetrically on
the guaranteed energy it delivers. Let
\begin{equation}
\ell(u,a) = -\,\mathbb{E}\big[U_i\big] = \mathbb{E}\big[\Pi\,s_i + \kappa\,D^g_i - b_a\big]
\label{eq:ell}
\end{equation}
be the \emph{single-participant loss} of a unit in state $u$ declaring item $a$ while every other
unit abstains. Table~\ref{tab:ell} gives its four values; all are positive, which is why zero
participation is strict.

\begin{table}[t]
\caption{Single-participant loss $\ell$, in \$ per event day.}
\label{tab:ell}
\centering
\begin{tabular}{llrr}
\toprule
state & item & $\ell$ & standard error\\
\midrule
normal   & conservative & $0.003104$ & $0.000000$\\
stressed & conservative & $0.003325$ & $0.000221$\\
normal   & aggressive   & $0.048305$ & $0.004916$\\
stressed & aggressive   & $0.199928$ & $0.009958$\\
\bottomrule
\end{tabular}
\end{table}

The corrective transfer is paid only where the aggregator is short of declared capability. Let
$Q_{-i}$ be the total capability declared by units other than $i$, and $\bar{Q}=9.0$~kW the
capability target, $0.60$ of the full-participation capability $Nq_A=15.0$~kW. The transfer is
\begin{equation}
R_i = \mathbf{1}\{i \text{ participates}\}\; s_0\, f\big(Q_{-i}\big),
\label{eq:transfer}
\end{equation}
paid on top of \eqref{eq:payment}, so a participating unit's realized settlement is $U_i+R_i$.
The scale $s_0$ sets how large the payment can be; the \emph{decay structure} $f$ sets how it
falls as the other units declare more. We restrict $f$ by four conditions:
\begin{assumption}\label{as:class}
\begin{enumerate}
\item[\textup{(A0)}] $R_i=0$ for a unit that abstains, whatever the others declare;
\item[\textup{(A1)}] conditional on participating, $R_i$ does not depend on which item $i$
declared, since $f$ reads $Q_{-i}$ alone;
\item[\textup{(A2)}] $f(0)=1$, so a unit participating while all others abstain receives $s_0$;
\item[\textup{(A3)}] $f$ is non-increasing and $f(Q_{-i})=0$ once $Q_{-i}\ge\bar{Q}$.
\end{enumerate}
\end{assumption}

Equation~\eqref{eq:transfer} is the mechanism's joint coupling in explicit form: conditional on
participating, unit $i$'s transfer is a function of the other units' declarations and of nothing
else it controls. Coupling enters a second time through the dispatch, which is solved from the
whole profile and therefore fixes $x_i$---and hence $D^g_i$ and $s_i$ in
\eqref{eq:utility}---from $\mathbf{a}$ rather than from $a_i$.

\subsection{Invariance and the static criteria}

\begin{theorem}[Invariance]\label{th:inv}
Under Assumption~\ref{as:class} and $(N-1)q_C\ge\bar{Q}$, adding $R$ to the settlement leaves
every allocation, every payment and every contract-selection margin at the intended operating
point unchanged, for any $s_0$.
\end{theorem}

\noindent\emph{Proof:} See Appendix~\ref{app:inv}.

The invariance is exact, and it is what makes $s_0$ a free parameter: no choice of it can disturb
the reporting calibration. Since $f(0)=1$ for every admissible structure, the transfer at zero
participation is exactly $s_0$, so each single-participant loss is a constant an entry
requirement must clear:
\begin{equation}
\begin{aligned}
s_0^{\mathrm{elim}} &= \min_{(u,a)} \ell = 0.003104, \\
s_0^{\mathrm{entry,truthful}} &= \max_{\text{truthful }(u,a)} \ell = 0.048305, \\
s_0^{\mathrm{entry,any}} &= \max_{(u,a)} \ell = 0.199928 .
\end{aligned}
\label{eq:thresholds}
\end{equation}

\begin{corollary}\label{th:cor}
For $s_0 > s_0^{\mathrm{elim}}$, zero participation is not an equilibrium of the corrected
mechanism, and the reporting constants are unchanged.
\end{corollary}

\noindent\emph{Proof:} See Appendix~\ref{app:cor}.

Every $f$ satisfying Assumption~\ref{as:class} gives the same three thresholds, because all are
evaluated at $Q_{-i}=0$ where $f=1$. The class is therefore not pinned down by any of them. We
carry two members forward:
\begin{equation}
f^{\mathrm{lin}}(Q_{-i}) = \Big(1-\tfrac{Q_{-i}}{\bar{Q}}\Big)_{+},
\qquad
f^{\mathrm{thr}}(Q_{-i}) = \mathbf{1}\big\{Q_{-i}<\bar{Q}\big\}.
\label{eq:two}
\end{equation}

\begin{definition}[On-path payoff equivalence]\label{def:equiv}
Two corrective transfers are \emph{on-path payoff equivalent} when they agree on (i) the
allocation at the intended profile, (ii) the payments there, (iii) every contract-selection
margin, (iv) $s_0^{\mathrm{elim}}$, and (v) $s_0^{\mathrm{entry,any}}$.
\end{definition}

By Theorem~\ref{th:inv} and \eqref{eq:thresholds}, the two structures of \eqref{eq:two} are
on-path payoff equivalent for every $s_0$. This is an equivalence over five named criteria, not
an identity of the two mechanisms: their payoffs differ at the \emph{partial-participation}
states, which no criterion in Definition~\ref{def:equiv} inspects, since every one is evaluated
either at the intended profile or at zero participation, where the two agree by construction.

\subsection{Adaptive participants: learner state and the update}

Units do not solve the mechanism. Each unit keeps a running mean $\hat{u}_{i,t}(a)$ of the
realized settlement of each declaration $a$ \emph{it has actually made}, separately for each of
its two flexibility states $t$, together with a visit count $n_{i,t}(a)$. We call the pair
$(\hat{u}_i,n_i)$ unit $i$'s \emph{learner state}; it is private to the unit. Declarations are
drawn from the logit rule
\begin{equation}
\Pr\big[a_{i,t}=a\big]
= \frac{\exp\big(\beta\,\hat{u}_{i,t}(a)\big)}{\sum_{a'}\exp\big(\beta\,\hat{u}_{i,t}(a')\big)},
\label{eq:logit}
\end{equation}
at $\beta=4$, and only the declaration actually made is updated, toward the settlement actually
earned:
\begin{equation}
\hat{u}_{i,t+1}(a)=
\begin{cases}
\hat{u}_{i,t}(a) + \dfrac{w_i(r)-\hat{u}_{i,t}(a)}{n_{i,t}(a)}, & a=a_{i,t},\\[6pt]
\hat{u}_{i,t}(a), & \text{otherwise,}
\end{cases}
\label{eq:update}
\end{equation}
where $w_i(r)=U_i+R_i$ is unit $i$'s realized settlement in round $r$; we write $w$ rather than
$x$ because $x_i$ already denotes commanded discharge power in \eqref{eq:blocks}. Equation
\eqref{eq:update} is where the feedback structure lives: the counterfactual is never supplied,
so learning is bandit. A \emph{round} is one participation opportunity drawn from a fixed
library of $400$ event days; $8000$ rounds is twenty passes through the library.

\begin{algorithm}[t]
\caption{Experience-based participation learning}
\label{alg:learn}
\begin{algorithmic}[1]
\State $\hat{u}(i,a,t)\!\gets\!0$, $n(i,a,t)\!\gets\!0$ for participating $a$;
       $\hat{u}(i,\mathrm{out},t)\!\gets\!\hat{u}^{(0)}$, $n(i,\mathrm{out},t)\!\gets\!n_0$
\For{round $r=0,\dots,T-1$}
  \State draw the event day and each unit's flexibility state $t(i)$
  \State each unit draws $a(i)$ by \eqref{eq:logit}
  \State the aggregator solves the dispatch from the joint profile
  \State delivery is executed and measured by \eqref{eq:meas}; settle by
         \eqref{eq:payment}--\eqref{eq:transfer}
  \State unit $i$ observes its own realized settlement $w_i$, and nothing else
  \State update $n$ and $\hat{u}$ by \eqref{eq:update}
\EndFor
\end{algorithmic}
\end{algorithm}

Every unit starts each flexibility state with $\hat{u}^{(0)}=0.20$ and pseudo-count $n_0=2000$ on
abstaining, and with $\hat{u}=0$, $n=0$ on both participating declarations. The pseudo-count
makes this a \emph{basin} rather than a point: abstaining realizes exactly zero, so the
abstention estimate decays as
\begin{equation}
\hat{u}^{(k)}(\mathrm{out}) = \frac{n_0\,\hat{u}^{(0)}}{n_0+k},
\label{eq:decay}
\end{equation}
with $k$ counting abstention draws in that unit's state. Over $8000$ rounds a state accumulates
about $1531$ such draws under the linear structure, so \eqref{eq:decay} reaches about $0.11$. The
incumbent persists because the prior is heavy, not because abstention pays.

\subsection{The two dependencies this paper turns on}
\label{sec:deps}

Two consequences of \eqref{eq:transfer}, \eqref{eq:logit} and \eqref{eq:update} are used later,
and are stated here so that they are read from the mechanism rather than asserted about the
implementation.

\emph{Across units, within a round.} By \eqref{eq:transfer}, unit $i$'s transfer is a function of
$Q_{-i}$. A change in any single unit $j$'s declaration changes $Q_{-i}$ simultaneously for every
$i\ne j$, and can therefore change the realized settlement of every other participating unit in
that round. Under $f^{\mathrm{lin}}$, which is strictly decreasing on $[0,\bar{Q})$, any such
change moves $R_i$ wherever $Q_{-i}$ lies in that range; under $f^{\mathrm{thr}}$, which is
piecewise constant, it does so only when the change crosses $\bar{Q}$. The dispatch channel can
act independently of either.

\emph{Across rounds, within a unit.} By \eqref{eq:update} the realized settlement is absorbed
into $\hat{u}_{i,t}$, which is persistent, and by \eqref{eq:logit} the next declaration is drawn
from a distribution determined by $\hat{u}_{i,t}$. A difference in a realized settlement
therefore does not dissipate: it is retained in the estimate and can change subsequent
declarations. Because \eqref{eq:logit} is full-support, the change is one of probability; with
the random draws held fixed it can change the realized declaration itself.

Neither statement is a theorem, and neither is claimed as one. Both are direct readings of
\eqref{eq:transfer}, \eqref{eq:logit} and \eqref{eq:update}.

% ============================================================ IV
\section{Distributed Realization of the Joint-Coupled Mechanism}
\label{sec:arch}

Section~\ref{sec:model} specifies the mechanism as a single evolving system. This section
decomposes it into independent processes without altering any of its definitions, and identifies
where that decomposition places the mechanism's information dependencies at risk. The purpose is
not to describe a network deployment but to make explicit which scientific quantities must cross
a process boundary, who is entitled to hold each of them, and at what instant a round becomes
semantically complete.

\subsection{Semantic objects and their ownership}

\emph{Private to participant $i$.} The estimate $\hat{u}_i$ and the visit count $n_i$. These are
never transmitted. No other participant sees them, and the aggregator does not hold, mirror, or
reconstruct them.

\emph{Private to the aggregator.} The payoff table and the operating calibration of
Section~\ref{sec:model}. No participant sees them, and no participant can compute its own
settlement.

\emph{Crossing a process boundary.} Exactly three classes of object, and no others: the
\emph{round state}, produced by the aggregator and consumed by every participant; the
\emph{declaration} of participant $i$, produced by $i$ and consumed by the aggregator; and the
\emph{participant-specific settlement}, produced by the aggregator and consumed by $i$ alone.

Two derived quantities never cross a boundary. The \emph{joint profile} is constructed at the
aggregator from independently produced declarations and is not published. The \emph{learning
update} is performed at each participant on its own state after its own settlement arrives. This
is the whole of the shared state; the decomposition adds nothing to Section~\ref{sec:model} and
withholds nothing from it.

\subsection{The distributed round}

One round proceeds in the mechanism's dependency order. Each arrow in Fig.~\ref{fig:arch} is a
boundary crossing; the intervening steps are local. The aggregator publishes the round state;
each participant retrieves it, reads its own type, and draws a declaration from \eqref{eq:logit}
applied to its private estimate, using its own random stream; each participant publishes its
declaration; the aggregator retrieves all $N$ declarations and assembles the joint profile; the
aggregator computes the settlement vector and publishes $N$ participant-specific settlements;
each participant retrieves its own and applies \eqref{eq:update}. The round is semantically
complete when every participant has applied its update.

The assembly step is a barrier, and it is required by the mechanism rather than chosen for
convenience: settlement is a function of the profile, so a profile assembled from a proper subset
of the declarations is a different profile, and can yield different settlements for participants
other than those omitted.

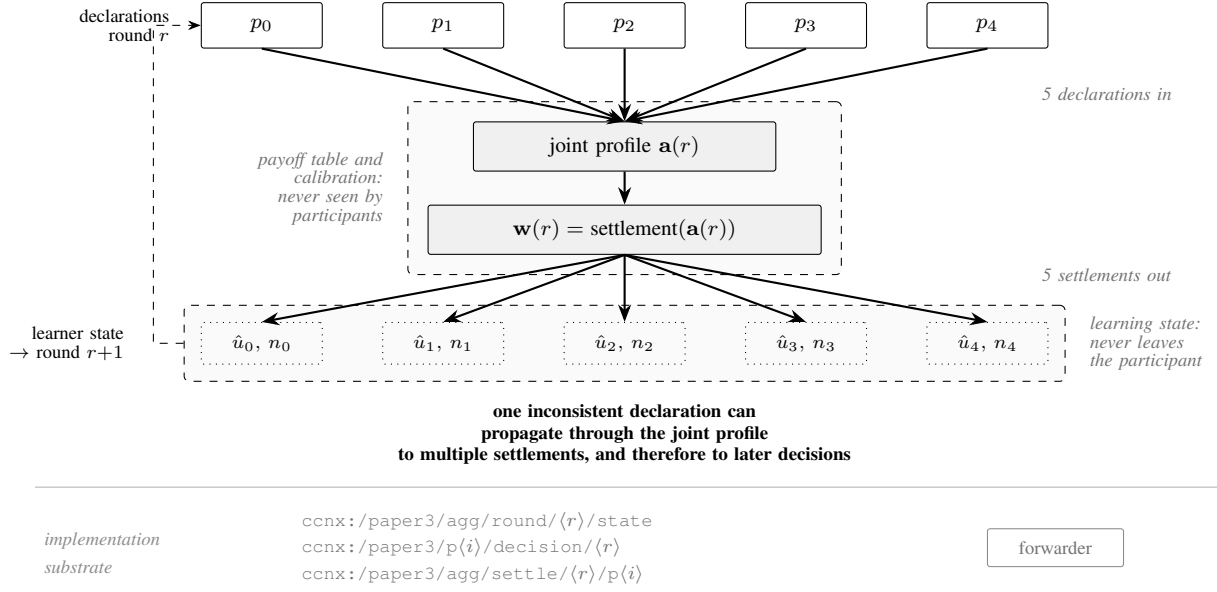
\begin{figure*}[t]
\centering
\begin{tikzpicture}[>=Stealth,
  pbox/.style={draw,rounded corners=1pt,minimum width=16mm,minimum height=6mm,font=\footnotesize},
  cart/.style={draw,dotted,rounded corners=1pt,minimum width=16mm,minimum height=5.5mm,font=\scriptsize},
  lbl/.style={font=\scriptsize},
  sub/.style={font=\scriptsize\itshape,text=black!55}]

% ---- rows -------------------------------------------------------------
\foreach \i in {0,...,4}{
  \node[pbox] (p\i) at ({\i*24mm},0) {$p_\i$};
  \node[cart] (c\i) at ({\i*24mm},-42mm) {$\hat{u}_\i,\,n_\i$};
}
\node[lbl,anchor=east,align=right] at (-11mm,0) {declarations\\[-1pt]round $r$};
\node[lbl,anchor=east,align=right] at (-17mm,-42mm) {learner state\\[-1pt]$\rightarrow$ round $r{+}1$};

\node[draw,fill=black!6,rounded corners=1pt,minimum width=40mm,minimum height=6.5mm,
      font=\footnotesize] (jp) at (48mm,-16mm) {joint profile $\mathbf{a}(r)$};
\node[draw,fill=black!6,rounded corners=1pt,minimum width=52mm,minimum height=6.5mm,
      font=\footnotesize] (st) at (48mm,-27mm) {$\mathbf{w}(r)=\text{settlement}(\mathbf{a}(r))$};

% ---- ownership boundaries (drawn first, behind everything) -------------
\begin{scope}[on background layer]
  \node[draw,dashed,fill=black!2,rounded corners=2pt,inner sep=2.6mm,fit=(jp)(st)] (aggb) {};
  \node[draw,dashed,fill=black!2,rounded corners=2pt,inner sep=2.2mm,fit=(c0)(c4)] (parb) {};
\end{scope}

% ---- flows ------------------------------------------------------------
\foreach \i in {0,...,4}{ \draw[->,thick] (p\i.south) -- (jp.north); }
\draw[->,thick] (jp) -- (st);
\foreach \i in {0,...,4}{ \draw[->,thick] (st.south) -- (c\i.north); }
\draw[->,dashed] (parb.west) -- ++(-4mm,0) |- (p0.west);

\node[sub,anchor=west] at (102mm,-9mm)  {5 declarations in};
\node[sub,anchor=west] at (102mm,-33mm) {5 settlements out};
\node[sub,anchor=east,align=right] at ($(aggb.west)+(-2mm,0)$)
  {payoff table and\\[-1pt]calibration:\\[-1pt]never seen by\\[-1pt]participants};
\node[sub,anchor=west,align=left] at ($(parb.east)+(2mm,0)$)
  {learning state:\\[-1pt]never leaves\\[-1pt]the participant};

\node[font=\scriptsize\bfseries,align=center,text width=70mm] at (48mm,-54mm)
  {one inconsistent declaration can propagate through the joint profile\\
   to multiple settlements, and therefore to later decisions};

% ---- substrate band ---------------------------------------------------
\draw[black!30] (-30mm,-61mm) -- (128mm,-61mm);
\node[sub,anchor=west,text=black!45] at (-30mm,-68mm) {implementation};
\node[sub,anchor=west,text=black!45] at (-30mm,-71.5mm) {substrate};
\node[font=\ttfamily\scriptsize,text=black!55,anchor=west] at (4mm,-65.5mm)
  {ccnx:/paper3/agg/round/$\langle r\rangle$/state};
\node[font=\ttfamily\scriptsize,text=black!55,anchor=west] at (4mm,-69mm)
  {ccnx:/paper3/p$\langle i\rangle$/decision/$\langle r\rangle$};
\node[font=\ttfamily\scriptsize,text=black!55,anchor=west] at (4mm,-72.5mm)
  {ccnx:/paper3/agg/settle/$\langle r\rangle$/p$\langle i\rangle$};
\node[draw,black!45,text=black!55,rounded corners=1pt,font=\scriptsize,
      minimum width=18mm,minimum height=5mm] at (105mm,-69mm) {forwarder};
\end{tikzpicture}
\caption{Distributed realization of the mechanism. Each participant holds its own estimate and
count, which never leave it; the aggregator holds the payoff table and calibration, which no
participant sees. Declarations produced independently at five sites are assembled into a single
joint profile, and that profile determines every participant's settlement---so one inconsistent
declaration can propagate to multiple settlements and, through the learning update of
\eqref{eq:update}, to later decisions. The lower band shows the named-data exchanges and the
forwarder used to instantiate this decomposition; the substrate carries the semantics but is not
the subject of the paper.}
\label{fig:arch}
\end{figure*}

\subsection{Naming}

Each of the three object classes is carried as named data. A name encodes what the object is,
which round it belongs to, and---for declarations and settlements---which participant it
concerns; the three names appear in Fig.~\ref{fig:arch}. Publisher identity is the routing
prefix, so forwarding requires $N+1$ static entries and no routing protocol. This naming is a
direct transcription of the ownership structure above; it is not offered as a contribution, and
nothing in the mechanism depends on it. Any substrate able to carry the same three object classes
with the same round and identity qualifiers would serve.

\subsection{Why the decomposition places the science at risk}
\label{sec:risk}

A declaration is not an independent input to a per-participant calculation. The aggregator
combines declarations produced independently, at different sites, into one profile, and that
profile determines every participant's settlement. Four properties of the assembly are therefore
load-bearing: the round each declaration belongs to; the participant each declaration came from;
the membership of the assembled set; and the recipient of each settlement. If any of the four is
inconsistent, the mechanism has not been executed, even though every participant behaved
correctly in isolation.

The consequence is not local. Because settlement is indexed by the joint profile and depends
again on the other participants' declarations through the participation transfer, a single
misattributed or omitted declaration can change the settlement of multiple participants in that
round---the transfer each participant realizes is a function of the capability the others
declared, so one altered declaration enters every other participant's transfer at once. Whether a given participant's numerical settlement moves depends on the decay structure and the
operating region, as Section~\ref{sec:deps} sets out, but the discrepancy is not confined to the
participant that produced it.

A second amplification follows. Each settlement is absorbed into a running mean with a visit
count, so a discrepancy does not remain confined to the round in which it occurred: it becomes
part of a persistent estimate. The next declaration is drawn from a policy over that estimate, so
the discrepancy can re-emerge as a change in behavior---for the affected participant directly,
and for other participants indirectly, because their next settlements are computed from the
profile that behavior helps determine.

A single inconsistent value at one round can therefore alter the trajectory of the whole
population for the remainder of the horizon. None of this is visible at the level of message
transfer alone. Each of the four properties above can be violated while every message is
delivered exactly once, in order, with no retransmission and no loss.

\subsection{Experimental substrate}

The realization used for validation runs the aggregator and the $N$ participants as independent
operating-system processes in separate network namespaces, communicating through real forwarders
using Interest/Data exchange, with application traffic traversing an intermediate forwarder. The
implementation contains no non-CCN communication path: there is no filesystem channel, no shared
object and no direct call between the agents, so every boundary crossing in
Fig.~\ref{fig:arch} is a network exchange. We propose no CCN protocol, claim no advantage over
other substrates, and treat no forwarder behavior as a scientific result. Using it makes the validation of Section~\ref{sec:fid} a measurement on a running distributed
system rather than an argument about a design.

\subsection{From realization to criterion}

Delivery is necessary---without it the round does not close, the profile is undefined and no
settlement exists. But Section~\ref{sec:risk} establishes that delivery is not sufficient: the
properties on which the mechanism's correctness depends are invisible to any transport-level
measurement, and a violation of any of them can propagate through settlement to multiple
participants and through learning to later rounds. Correctness of the realization must therefore
be assessed at the semantic levels the mechanism itself induces. Section~\ref{sec:fid} defines
those levels and the criterion applied at each.

% ============================================================ V
\section{Semantic-Fidelity Validation Methodology}
\label{sec:fid}

\subsection{Fidelity is not an output comparison}

The mechanism does not compute a value at a single site; it evolves a state that no single node
holds in full.
Participant $i$ maintains its own estimate $\hat{u}_i$ and count $n_i$, which the aggregator
never observes, while the aggregator holds the payoff table and calibration, which no participant
observes. The property the mechanism is about---how the population's private preferences evolve out of the collapse initialization---is a property
of the joint evolution of these private states, not of
any message or any node. Two structural features make that evolution fragile under
decentralization, and both are the dependencies of Section~\ref{sec:deps}: settlement is jointly
coupled, so a discrepancy affecting one declaration can change the settlement of multiple
participants in the same round; and the coupling is carried forward by learning, so a discrepancy
is not dissipated but absorbed into persistent state.

Together these mean that a single mismatched value at round $r$ can alter the declarations of
multiple participants at round $r+1$ and at later rounds. Comparing terminal outcomes cannot
detect this, because distinct trajectories may terminate at the same profile.

\subsection{Definition}

Let $\mathrm{C}$ denote the centralized reference implementation and $\mathrm{D}$ the distributed
realization, initialized identically and driven by the same exogenous type draws and the same
per-participant random streams. For round $r$ write $\mathbf{a}(r)\in\{0,1,2\}^N$ for the joint
declaration profile and $\mathbf{w}(r)\in\mathbb{R}^N$ for the \emph{realized-settlement vector},
with $w_i(r)=U_i+R_i$ the quantity participant $i$ receives and learns from, as in
\eqref{eq:update}. Write $(\hat{u}_i(r),n_i(r))$ for participant $i$'s learner state after
that round's update.

\begin{definition}[Semantic fidelity over horizon $T$]\label{def:fid}
$\mathrm{D}$ is semantically faithful to $\mathrm{C}$ over horizon $T$ if, for every round $r<T$
and every participant $i$,
$\mathbf{a}^{\mathrm{D}}(r)=\mathbf{a}^{\mathrm{C}}(r)$,
$\mathbf{w}^{\mathrm{D}}(r)=\mathbf{w}^{\mathrm{C}}(r)$,
$\hat{u}^{\mathrm{D}}_i(r)=\hat{u}^{\mathrm{C}}_i(r)$ and
$n^{\mathrm{D}}_i(r)=n^{\mathrm{C}}_i(r)$, with equality holding exactly.
\end{definition}

We use \emph{semantic fidelity} rather than \emph{semantic preservation} because the latter
denotes a machine-checked guarantee in compiler verification~\cite{leroy09} and would suggest a
stronger claim than the empirical one made here. The vector $\mathbf{w}$ collects the
per-participant realized settlements of \eqref{eq:update}; it is written $\mathbf{w}$ rather
than $\mathbf{s}$ because $s_i$ already denotes verified shortfall in \eqref{eq:short-i}.

Definition~\ref{def:fid} is a conjunction over rounds, not a statement about round $T-1$. That
distinction is the methodological content: for a mechanism whose errors are both persistent and
contagious across agents, fidelity is a per-round invariant, and any criterion evaluated only at
the end of the horizon is insufficient by construction.

\subsection{Validation hierarchy}
\label{sec:levels}

For reporting we organize the comparison into five validation layers, following the dependency
chain. L1--L4 compare the semantic objects of Definition~\ref{def:fid} directly; L5 is a derived
trajectory-level check included for interpretability. For each adjacent pair we give a concrete
way in which agreement at one layer can hold while agreement at the next fails. The layers are an
interpretable decomposition of what a round touches, not five independent proofs and not a proved
hierarchy: the negative control of Section~\ref{sec:negctl} is detected by four of them in the
same round.

\textbf{L1---Declaration fidelity.} The declaration each participant selects in each round. A
mismatch means the participant's policy or its view of the round differs. L1 does not imply L2:
correct individual declarations can still be assembled into an incorrect profile if a stale or
mis-attributed declaration is admitted.

\textbf{L2---Joint-profile fidelity.} The profile the aggregator assembles before settling. A
mismatch means the population's realized joint declaration differs, which can change the
settlement of multiple participants even where all declarations were individually correct. L2
does not imply L3: a correct profile can still be settled against the wrong table entry or the
wrong participant type index.

\textbf{L3---Settlement fidelity.} The per-participant realized-settlement vector. A mismatch
means at least one participant is paid a value the mechanism does not specify. L3 does not imply
L4: a correct settlement can be applied to the wrong (type, declaration) counter.

\textbf{L4---Learner-state fidelity.} Each participant's estimate and visit count after the
update. A mismatch means the divergence has entered persistent state and will be re-expressed in
later rounds. L4 at round $r$ does not imply L4 at $r+1$; only closure over all rounds gives
trajectory fidelity.

\textbf{L5---Trajectory and transition fidelity.} The sequence of coarse preference levels over
the horizon, and within it the round of first passage from $m=0$ to $m\ge1$.

\subsection{Why successful message delivery is not sufficient}
\label{sec:transport}

Transport-level success establishes that bytes arrived; it does not establish that the mechanism
was executed. Under the architecture of Section~\ref{sec:arch}, each of the following is
consistent with complete delivery, no retransmission and no timeout, and each violates
Definition~\ref{def:fid}: a consumer admits information belonging to a different round; a
declaration is attributed to the wrong participant when the profile is assembled; the round is
closed on an inconsistent set of declarations; a participant-specific settlement is delivered to
a different participant; a correct settlement is applied to the wrong entry of a participant's
local estimate. None of these is detectable by counting messages, and none requires a network fault.
\emph{Transport success and semantic fidelity are independent properties, and only the second is
evidence that the mechanism was executed.}

\subsection{Exact equality rather than statistical agreement}

Because both implementations are driven by identical exogenous draws, the mechanism is
deterministic given those draws, and the comparison criterion is equality rather than
distributional agreement. The substance of the comparison is that it is made directly, per
round, on four quantities---each participant's declaration, the joint profile the aggregator
assembles, the participant-specific settlement, and each participant's learner state after its
update---rather than on any summary of them. Equality rather than agreement within a band is
used because a statistical criterion would admit each of the failure modes of
Section~\ref{sec:transport} at low rates; that the comparison carries no tolerance to select is
a property of the criterion rather than its main content. No numerical tolerance is used
anywhere in the comparison reported in Section~\ref{sec:dist}.

Under a correct implementation with fixed randomness, exact equality is the expected outcome, and
we do not present it as surprising. What the comparison contributes is where it is taken: on the
mechanism-level state whose preservation is required for subsequent adaptive behavior, rather
than inferred from successful transport or from a terminal outcome. Section~\ref{sec:negctl}
reports a controlled negative control showing that an execution in which every message is
delivered successfully can nevertheless violate these comparisons. Exact equality over one
configuration is not a claim of equivalence in general: it establishes fidelity for the canonical
operating point, one seed, $N=5$, and the conditions under which the comparison was run.

\subsection{The qualitative transition as an integrity check}
\label{sec:integrity}

If Definition~\ref{def:fid} holds over the horizon, agreement on the transition round follows; it
is not independent evidence. Its role is interpretive. The transition is the threshold crossing
of a slowly-moving statistic late in the horizon, so any divergence at any earlier round would
displace it. Reporting the round at which it occurs therefore compresses a full-horizon
comparison into a single checkable number, and it demonstrates that the validated interval
contains the late coarse-preference first passage rather than only a long quiet prefix.

% ============================================================ VI
\section{Adaptive-Participation Results}
\label{sec:results}

Unless stated otherwise: $s_0=s_0^{\mathrm{entry,any}}=0.199928$, $\beta=4$, $8000$ rounds, $96$
seeds, initialization as in Section~\ref{sec:model}. Settlements are the \emph{realized} episode
values: a round draws one episode seed and every unit is paid what it earned in that episode, so
the day's common shock is intact.

\subsection{The two structures separate}

Convergence to full participation is the event that $\arg\max_a \hat{u}_{i,t}(a)$ is the truthful
item for every $(i,t)$ at the end of the run; rates carry $95\%$ Wilson intervals. From the
collapse initialization, convergence within the $8000$-round canonical horizon is $0.00$
$[0.000,0.038]$ with no corrective transfer \emph{and} under the linear structure, and $1.00$
$[0.962,1.000]$ under the thresholded one (Table~\ref{tab:head}). These are finite-horizon
rates; Section~\ref{sec:feedback} reports the long-horizon control. From a random start all three are $1.00$ $[0.962,1.000]$: the intervals
separate completely from the collapse initialization and coincide exactly away from it, so the difference is not a
general performance gap---it is specific to departure from that basin.

Two transfers that Theorem~\ref{th:inv} and \eqref{eq:thresholds} make indistinguishable under
all five static criteria therefore lead adaptive units to different outcomes within the
canonical horizon.

Three distinct notions appear below and are not interchangeable: the \emph{full-participation
convergence criterion} just stated, which is the basis of every rate in this section and does
not involve $m$; the \emph{coarse preference level} $m(r)$ of \eqref{eq:mdef}, a trajectory
statistic; and the \emph{realized declarations} actually made in a round, which are what the
settlement is computed from.

\begin{table}[t]
\caption{On-path payoff equivalent, different under learning within the canonical horizon.
Convergence is to full participation---$\arg\max_a\hat{u}_{i,t}(a)$ the truthful item for every
$(i,t)$---from the collapse initialization at $s_0=s_0^{\mathrm{entry,any}}$, $96$ seeds,
$8000$ rounds. Rates are finite-horizon; see Section~\ref{sec:feedback}.}
\label{tab:head}
\centering
\begin{tabular}{lll}
\toprule
settlement & static criteria & convergence\\
\midrule
no transfer & --- & $0.00$\\
linear & as no transfer & $0.00$\\
thresholded & as no transfer, and as linear & $1.00$\\
\bottomrule
\end{tabular}
\end{table}

\subsection{Where the difference lives}
\label{sec:where}

Table~\ref{tab:rungs} gives the partial-participation payoffs $\Delta^{\mathrm{join}}(j)$, the
expected payoff to a unit joining when $j$ other units already participate. Without a corrective
transfer only the $j=0$ row is negative: both flexibility states lose by joining alone and gain
at every later level. That single row \emph{is} the zero-participation equilibrium---the barrier
is the first step and nothing else.

\begin{table}[t]
\caption{Expected payoff $\Delta^{\mathrm{join}}(j)$ to a unit joining when $j$ other units
already participate, \$ per event day, at $s_0=s_0^{\mathrm{entry,any}}$. The index $j$ is
distinct from the coarse preference level $m(r)$ of \eqref{eq:mdef}. Normal state; the stressed
state has the same sign pattern, its no-transfer column running from $-0.003253$ at $j=0$ to
$+0.041203$ at $j=4$.}
\label{tab:rungs}
\centering
\begin{tabular}{lrrr}
\toprule
$j$ other units & no transfer & linear & thresholded\\
\midrule
$0$ & $\mathbf{-0.048927}$ & $0.151001$ & $0.151001$\\
$1$ & $0.066963$ & $0.205802$ & $\mathbf{0.266892}$\\
$2$ & $0.066042$ & $0.143792$ & $\mathbf{0.265971}$\\
$3$ & $0.062823$ & $0.079484$ & $\mathbf{0.237761}$\\
$4$ & $0.061983$ & $0.061983$ & $0.061983$\\
\bottomrule
\end{tabular}
\end{table}

\emph{Why the $j=0$ entry differs from Table~\ref{tab:ell}.} The no-transfer payoff at $j=0$ in
Table~\ref{tab:rungs} is $-0.048927$, whereas the corresponding single-participant loss in
Table~\ref{tab:ell} is $0.048305$. The two are measured on different estimands by design.
Equation~\eqref{eq:ell} is evaluated on the one unit that every deviation counterfactual is also
evaluated on, matched on the event-day seed, because a deviation and the truthful baseline it is
compared against must be read on the same unit. Table~\ref{tab:rungs} instead pools over all five
units, which is the average a joining unit faces and therefore the quantity the learning process
sees. The two conventions differ in the fourth decimal; recomputed on the designated unit, the
$j=0$ row of Table~\ref{tab:rungs} reproduces the thresholds of \eqref{eq:thresholds} to
$10^{-8}$. They are not two estimates of a single quantity.

The two structures agree at both ends and differ only in between, which is where the statistic
used below is taken:
\begin{equation}
\min_{1\le j\le N-2}\Delta^{\mathrm{join}}(j)=
\begin{cases}
0.079484, & f=f^{\mathrm{lin}},\\
0.237761, & f=f^{\mathrm{thr}}.
\end{cases}
\label{eq:minrung}
\end{equation}
At $j=0$ both pay $s_0$ because $f(0)=1$; at $j=N-1$ both switch off because the others already
meet $\bar{Q}$. The static criteria are evaluated at exactly those two points, which is why they
cannot separate \eqref{eq:minrung}.

Neither structure is increasing in $j$ throughout: the thresholded payoffs \emph{ease} from
$0.266892$ at $j=1$ to $0.237761$ at $j=3$. What distinguishes the two is the \emph{level} of the
partial-participation payoffs, not their monotonicity---a thresholded structure does not work
because it rewards later entrants more, but because it keeps every intermediate payoff high.

That level must be measured against the learner's own estimate of abstaining. With the margin
$\delta(j,k)=\Delta^{\mathrm{join}}(j)-\hat{u}^{(k)}(\mathrm{out})$ from \eqref{eq:decay}, at
$k=0$ the linear structure clears the incumbent by $0.005802$ at $j=1$ and is \emph{below} it by
$0.056208$ at $j=2$, whereas the thresholded structure clears it by $0.066892$, $0.065971$ and
$0.037761$ at $j=1,2,3$. A negative margin does not say a state cannot be reached; it says what a
unit earns there is weak against the estimate it already holds, so observations there do not
sustain the participating estimate. That is a statement about \emph{retention}.

\subsection{Arrival, retention, and propagation}
\label{sec:retention}

These are three events, not one, and the statistic used to separate them must be stated
precisely. Let
\begin{equation}
k(r) = \textstyle\sum_{i,t}\mathbf{1}\big\{\arg\max_a \hat{u}_{i,t}(a)\neq\text{abstain}\big\}
\in\{0,\dots,2N\}
\label{eq:kdef}
\end{equation}
count the \emph{unit--flexibility-state learner cells} whose current $\arg\max$ favors a
participating declaration; each unit contributes $0$, $1$ or $2$ cells. The frozen analysis
tracks the \emph{coarse preference level}
\begin{equation}
m(r) = \lfloor k(r)/2 \rfloor \in \{0,\dots,N\}.
\label{eq:mdef}
\end{equation}

The integer division makes $m$ a coarse level derived from the $2N=10$ type-specific preference
cells, and not a count of units or of realized declarations. Its endpoints are asymmetric:
$m=5$ holds if and only if all ten cells favor participation, whereas $m=0$ admits $k=0$
\emph{or} $k=1$ and therefore does not imply that no cell, and still less that no unit, favors
participation---nor does it constrain the declarations actually made in that round. We attach no
stronger reading to the intermediate levels than \eqref{eq:mdef} supports.

Fig.~\ref{fig:escape}(a) reports, for each of the $48$ tracked seeds, whether a run has
\emph{reached} a given $m$-level by round $r$---the running maximum of $m$.

Every seed of both structures reaches $m=1$ within a median of one round, which
Table~\ref{tab:rungs} predicts, since both pay $0.151001$ at $j=0$. Beyond that they part: the
linear structure reaches $m=2$ in $33$ of $48$ seeds and $m=5$ in none within the horizon; the
thresholded structure reaches $m=2$ in all $48$ and $m=5$ in all $48$, within a median of $159$
rounds.

Reaching a level is not holding one, and that is where the finding is. Table~\ref{tab:dwell}
shows the linear structure spending $0.004$ of its rounds at $m=2$ and $0.000$ at $m=3$: it
reaches the second level and does not keep it. That is what a negative $\delta(2,k)$
describes---not a level the population cannot visit, but one whose payoff does not sustain the
participating estimate once visited. The thresholded structure, whose margins stay positive
throughout, spends $0.882$ of its rounds at $m=5$, which by \eqref{eq:mdef} is exactly the state
in which every unit favors participation in both flexibility states.

\begin{table}[t]
\caption{Share of rounds at each coarse preference level $m$ of \eqref{eq:mdef}, from the
collapse initialization, $s_0=s_0^{\mathrm{entry,any}}$, $48$ tracked seeds, $8000$ rounds.
Occupancy, not first passage; Fig.~\ref{fig:escape}(a) reports which $m$-levels were reached at
all. $m$ is a coarse level over the ten type-specific preference cells, not a count of
participating units or of realized declarations.}
\label{tab:dwell}
\centering
\begin{tabular}{lrrrrrr}
\toprule
 & $0$ & $1$ & $2$ & $3$ & $4$ & $5$\\
\midrule
linear & $\mathbf{0.887}$ & $0.108$ & $0.004$ & $0.000$ & $0.000$ & $0.000$\\
thresholded & $0.000$ & $0.001$ & $0.008$ & $0.035$ & $0.074$ & $\mathbf{0.882}$\\
\bottomrule
\end{tabular}
\end{table}

\begin{figure*}[t]
\centering
\includegraphics[width=\textwidth]{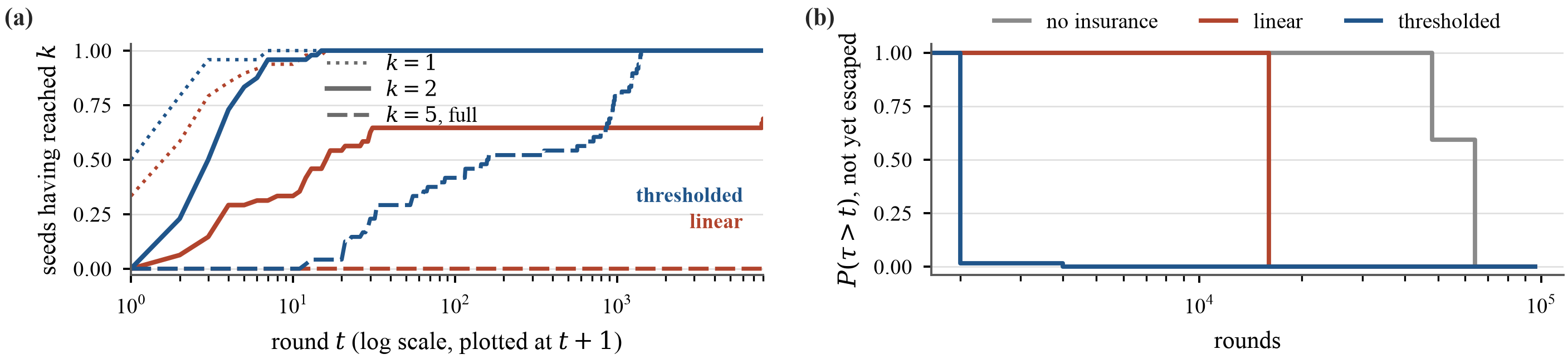}
\caption{(a) Fraction of the $48$ tracked seeds whose running maximum of the coarse preference
level $m$ of \eqref{eq:mdef} has reached a given $m$-level by round $r$, on a logarithmic round
axis,
under the realized estimator. Both structures reach $m=1$ rapidly; they separate beyond it, the
thresholded structure reaching $m=5$ in every tracked seed and the linear structure in none
within the horizon. (b) First-passage behavior out of the collapse initialization as a survival
curve at $\beta=4$ over $64$ seeds. The legend label \emph{no insurance} denotes the mechanism
with no corrective transfer.}
\label{fig:escape}
\end{figure*}

\subsection{The threshold ordering reverses}

Write $s_0^{\mathrm{L}}$ for the smallest scale on the swept grid at which the Wilson lower bound
on convergence reaches $0.9$ at an $8000$-round horizon. Against the static thresholds
\eqref{eq:thresholds},
\begin{equation}
\begin{aligned}
\text{thresholded:}\quad & s_0^{\mathrm{elim}} < s_0^{\mathrm{L}} < s_0^{\mathrm{entry,any}},
  && s_0^{\mathrm{L}}=0.12,\\
\text{linear:}\quad & s_0^{\mathrm{elim}} < s_0^{\mathrm{entry,any}} < s_0^{\mathrm{L}},
  && s_0^{\mathrm{L}}=0.35.
\end{aligned}
\label{eq:order}
\end{equation}
The thresholded structure reaches the population \emph{before} the scale is large enough to make
a single entrant whole; the linear structure makes a single entrant whole long before it reaches
the population. Within the canonical horizon the linear structure did not carry the population to full
participation even at $s_0^{\mathrm{entry,any}}$, about sixty-four times
$s_0^{\mathrm{elim}}$. The static constant is
therefore not a conservative proxy for the dynamic one.

\subsection{The separation is a property of the feedback structure}
\label{sec:feedback}

Fig.~\ref{fig:escape}(b) shows that what differs is the time a population takes, not what it can
eventually reach: at a $96\,000$-round horizon every condition converges, including the mechanism
with no corrective transfer. The comparison is repeated under regret
matching~\cite{hartmascolell00} on the same own-payoff estimator and in the rule's
full-information form, where every declaration's counterfactual payoff is supplied. Between the
two, convergence for the linear structure goes from $0.20$ $[0.13,0.29]$ to $1.00$, while the
thresholded structure is at $1.00$ either way. \emph{The separation closes entirely under full
information.} It is a statement about what an experience-based learner can discover, not about
the incentive as an object.

\subsection{How wide the window is}
\label{sec:window}

The separation is not an artifact of one initialization. Sweeping the incumbent value: at $0.10$
the linear structure has margin to spare and retains everyone; by $0.15$ its convergence rate has
fallen to $0.39$; by $0.20$, where $\delta(1,0)=0.005802$, it is zero. The lower edge lies between
$0.10$ and $0.15$, which is where the margin says it should.

\subsection{A one-parameter family} \label{sec:family}  Two structures are two points. Varying the profile along one dimension and nothing else,
\begin{equation}
f_\gamma\big(Q_{-i}\big)=\Big(1-\tfrac{Q_{-i}}{\bar{Q}}\Big)_{+}^{\gamma},
\label{eq:family}
\end{equation}
nine members are tested: $\gamma\in\{0.10,0.25,0.35,0.50,0.60,0.70,0.85,1.00\}$ together with the
thresholded structure, the limiting member as $\gamma\to 0$. Every member satisfies
Assumption~\ref{as:class}. Convergence is monotone in $\min_j\Delta^{\mathrm{join}}(j)$ across all
nine with no inversion: as the minimum falls from $0.237761$ to $0.079484$ the convergence rate
falls from $1.00$ to zero and never rises (Fig.~\ref{fig:gamma}). The transition is sharp---between $\gamma=0.5$ and $\gamma=0.6$---and sits at a minimum between
$0.104649$ and $0.115688$. The statistic ranks; it does not, by itself, locate.  This family
cannot identify \emph{which} summary of the partial-participation payoffs is responsible, and we
do not claim that it does. The family is pointwise monotone in $\gamma$: for every pair of
members, one member's intermediate payoffs dominate the other's at every $j$, without exception
across all thirty-six pairs. Any statistic that is monotone in those payoffs therefore induces the
same ranking---the uniform mean does, every non-negative weighting of the intermediate payoffs we
examined does, and so does $\gamma$ itself. What \eqref{eq:minrung} provides is a descriptor of
the partial-participation region that is computable before any simulation and that ranks these
structures correctly; separating it from a weighted-average alternative would require a family
whose members are not pointwise ordered, which we do not construct here.

\begin{figure}[t]
\centering
\includegraphics[width=\columnwidth]{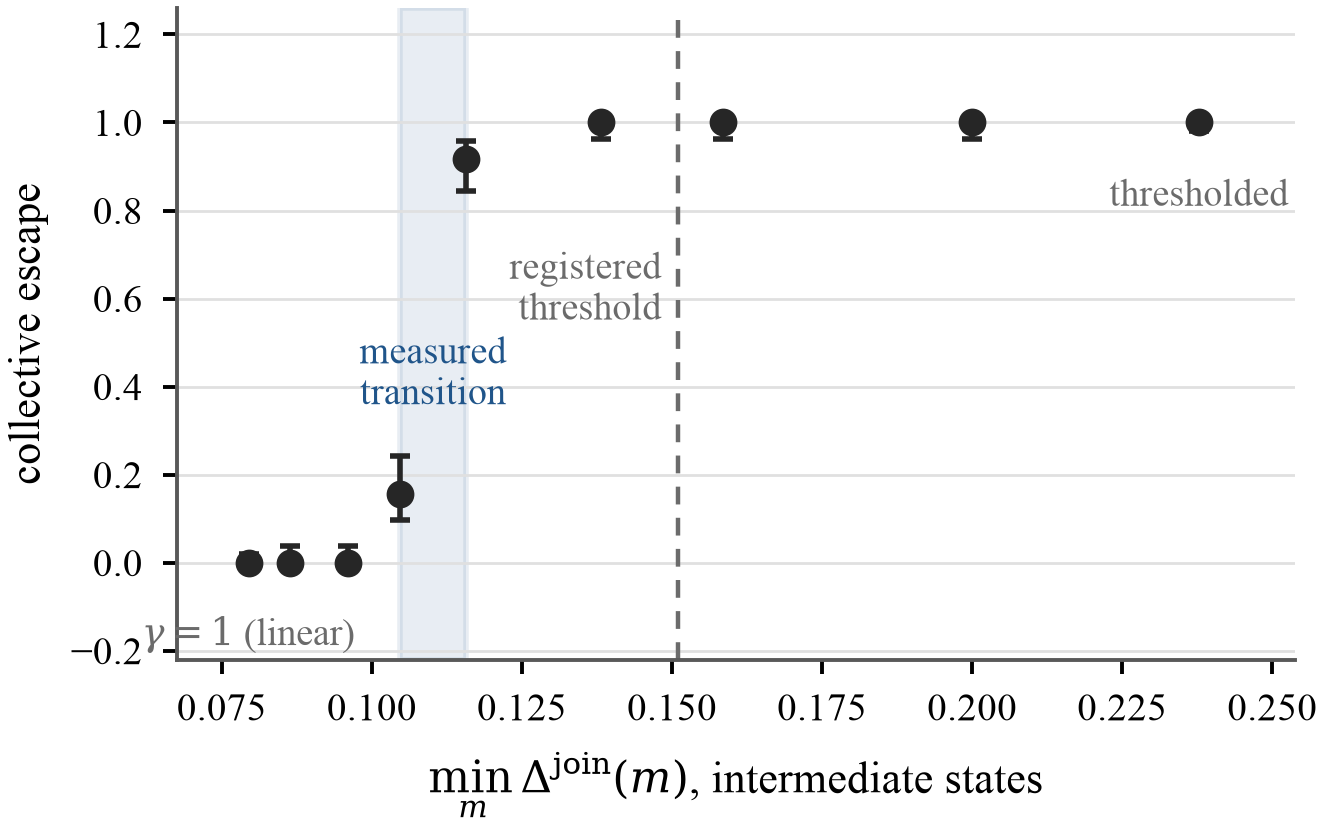}
\caption{Convergence to full participation within the canonical horizon against the minimum
partial-participation payoff $\min_j\Delta^{\mathrm{join}}(j)$, over the nine members of
\eqref{eq:family}. The ordering is monotone with no inversion. Each marker carries the $95\%$ Wilson interval of its convergence rate over the same $96$ seeds,
computed for all nine members by the same rule; the intervals show finite-sample uncertainty and
are not used to establish the ordering. The ordinate label \emph{collective escape} denotes the
same full-participation convergence criterion used in Table~\ref{tab:head}; in the abscissa
label, $m$ corresponds to the partial-participation index $j$ in Table~\ref{tab:rungs}.}
\label{fig:gamma}
\end{figure}

\subsection{What does not separate them, and budget feasibility}
\label{sec:noise}

Adding zero-mean noise to the realized settlement each unit observes and sweeping its standard
deviation over $[0.16,0.50]$---that is, from $0.96$ to $3.01$ times the natural payoff standard
deviation---degrades convergence for every condition, but by indistinguishable amounts.
Reward-noise tolerance is therefore not an instrument that distinguishes these two structures.
This sweep is computed under a second canonical construction, the \emph{profile-mean} estimator,
in which a unit is paid the mean settlement of its (profile, type) cell; the distributed
validation of Section~\ref{sec:dist} reproduces that same construction at zero added noise.

Writing $Q_t$ for the reduction delivered in round $t$ and $v$ for its value per kWh, the
aggregator's net cash flow is
\begin{equation}
\Pi^{\mathrm{agg}}_t = v\,Q_t - \sum_i \underbrace{\Big(b_i - \Pi\,s_{i,t} + \kappa\,D^z_{i,t}
 + s_0 f_\gamma\big(Q_{-i,t}\big)\Big)}_{P_{i,t}+R_{i,t}} .
\label{eq:budget}
\end{equation}
All nine members clear ex-post feasibility: $v^{\min}$ runs from $0.14079$ to $0.15942$, in every
case under a sixth of the shortfall price $\lambda^-=1.00$~\$/kWh, and $\Pi^{\mathrm{agg}}_t\ge0$
holds in every one of the $6\,912\,000$ rounds measured---$96$ seeds $\times$ $8000$ rounds
$\times$ nine members---at $v^{\mathrm{expost}}=0.30844$. The separation is consequently not an
artifact of aggregator deficit.

% ============================================================ VII
\section{Distributed Semantic-Fidelity Validation}
\label{sec:dist}

\subsection{Configuration}

The distributed realization of Section~\ref{sec:arch} and the centralized reference were executed
at the canonical operating point (\,$\Pi=1.3925$, $b_C=0.5849$, $b_A=0.6772$, linear transfer,
$s_0=0.19993$\,), with $N=5$, seed $12345$, collapse initialization, and a horizon of $T=8000$
rounds---the canonical horizon of Section~\ref{sec:results}. The two executions were driven by
the same exogenous type schedule and the same per-participant random substreams, so the mechanism
is deterministic given those draws.

\emph{Which canonical construction is reproduced.} Section~\ref{sec:results} evaluates the
mechanism under two canonical settlement constructions. Tables~\ref{tab:head}--\ref{tab:dwell} and
Fig.~\ref{fig:escape} use the \emph{realized} estimator, in which each round draws one episode
seed and every unit is paid what it earned on that day. The distributed validation reproduces the
\emph{profile-mean} estimator at zero added reward noise, in which a unit is paid the mean
settlement of its (profile, type) cell. Both use the same coarse preference level $m(r)$ of
\eqref{eq:mdef} and the same first-passage criterion $m\ge1$; they differ only in the settlement
that drives learning, and consequently in how quickly $m$ first reaches $1$. The two
first-passage results are therefore not comparable and are not in conflict:
Fig.~\ref{fig:escape}(a) reports first passage under the realized estimator, and the round-$7538$
first passage reported below is under the profile-mean estimator. Round $7538$ is not the
corresponding time in Fig.~\ref{fig:escape}(a).

The success criteria, the transition definition, and the classification of every anticipated
failure mode were registered before the run and were not modified afterwards.

\emph{What is compared.} The quantity settled, published, received and learned from is
$w_i=U_i+R_i$---the utility of \eqref{eq:utility} plus the corrective transfer of
\eqref{eq:transfer}---evaluated on the profile-mean cell for the round's joint profile and type
vector. It is neither the payment $P_i$ nor the transfer $R_i$ alone.

\subsection{Primary result}

At the canonical operating point, the distributed realization reproduced the centralized
reference over the complete $8000$-round horizon with exact equality at every fidelity level of
Section~\ref{sec:levels}. Across $8000$ rounds there were no joint-declaration mismatches, the
maximum absolute settlement difference was $0.0$, there were no coarse-preference-level mismatches,
the maximum absolute learner differences were $0.0$ for $\hat{u}$ and $0$ for $n$, and the
terminal $\arg\max$ was identical. The first passage of the frozen coarse preference statistic from $m=0$ to $m\ge1$ occurred at
round $7538$ in both executions. No numerical tolerance was introduced at any point in the
comparison. Table~\ref{tab:fidelity} states the evidence.

\subsection{Results by fidelity level}

\textbf{L1---Declaration fidelity.} Each participant's own record of the declaration it selected
was compared against the centralized reference for the same round and participant: $40\,000$
comparisons ($8000$ rounds $\times$ $5$ participants), zero mismatches. The same records were
compared independently against the declaration the aggregator admitted into that round's
profile---two accounts written by different processes on different hosts---again with zero
mismatches in $40\,000$ comparisons. This is a participant-side measurement and is independent of
the profile comparison that follows.

\textbf{L2---Joint-profile fidelity.} The profile assembled by the aggregator was identical to the
reference profile in all $8000$ rounds. Because settlement is indexed by the profile, this is the
level at which a single mis-attributed declaration could have changed the settlement of multiple
participants; it did not occur in any round.

\textbf{L3---Settlement fidelity.} The maximum absolute difference between the distributed and
reference realized-settlement vectors, over all $8000$ rounds and all five participants, was
$0.0$---exact equality, not agreement within a tolerance. Every participant-specific settlement
value was therefore exactly reproduced.

\textbf{L4---Learner-state fidelity.} Each participant's estimate and count after every one of its
$8000$ updates were compared entry by entry against the reference state after the same round:
$240\,000$ scalar comparisons for $\hat{u}$ and $240\,000$ for $n$, with zero mismatching entries
and maximum absolute differences of $0.0$ and $0$. The terminal $\arg\max$ was identical,
$[[0,0],[2,0],[2,0],[2,0],[0,0]]$, in both executions.

\textbf{L5---Trajectory and transition fidelity.} The coarse preference level $m(r)$ of
\eqref{eq:mdef} agreed at every round: zero mismatches over $8000$ rounds. Within that, the
first passage of $m$ from $0$ to $1$---which is neither convergence to full participation nor the
onset of participation---occurred at round $7538$ in both executions, and the $200$-round
verification window registered in advance (rounds $7538$--$7737$) contained no mismatch at any
level.

Definition~\ref{def:fid} is a conjunction over rounds; it is satisfied over the full canonical
horizon for this configuration.

\begin{table}[t]
\caption{Exact semantic-fidelity result, distributed versus centralized execution, canonical
operating point, profile-mean estimator, $N=5$, seed $12345$.}
\label{tab:fidelity}
\centering
\begin{tabular}{lr}
\toprule
validation horizon & $8000$ rounds\\
participants $N$ & $5$\\
comparison tolerance & \textbf{none (exact equality)}\\
\midrule
\multicolumn{2}{l}{\emph{Semantic fidelity}}\\
L1 participant declaration vs.\ centralized & $\mathbf{0}$ of $40\,000$\\
L1 participant declaration vs.\ aggregator & $\mathbf{0}$ of $40\,000$\\
L2 joint-profile mismatches & $\mathbf{0}$ of $8\,000$\\
L3 settlement mismatches & $\mathbf{0}$ of $40\,000$\\
L3 settlement max abs.\ difference & $\mathbf{0.0}$\\
L4 learner $\hat{u}$, per-round mismatching entries & $\mathbf{0}$ of $240\,000$\\
L4 learner $n$, per-round mismatching entries & $\mathbf{0}$ of $240\,000$\\
L4 terminal $\arg\max$ equality & \textbf{identical}\\
L5 coarse-preference-level mismatches & $\mathbf{0}$ of $8\,000$\\
L5 first-passage round, centralized & $\mathbf{7538}$\\
L5 first-passage round, distributed & $\mathbf{7538}$\\
\midrule
\multicolumn{2}{l}{\emph{Message plane} (not fidelity evidence)}\\
aggregator fetches, successful & $40\,000$\\
application-level retries & $0$\\
aborts & $0$\\
forwarders alive at completion & $6$ of $6$\\
elapsed / per round & $15.76$~h / $7.09$~s\\
\bottomrule
\end{tabular}
\end{table}

\subsection{The coarse-preference first passage}

The first passage of $m$ is the qualitative event in the frozen trajectory statistic, and it lies
late in the horizon. In the reference, $m$ remains $0$ through round $7537$---where $k=1$, the
settled $\arg\max$ is $[[0,0],[0,0],[0,0],[2,0],[0,0]]$ and the joint declaration is
$21011$---and reaches $m=1$ at round $7538$, where $k=2$ as participant $p_2$'s normal-state
$\arg\max$ moves from abstain to the aggressive item. The distributed execution reproduced the
same first passage at the same round, produced by the same participant's same cell change, and
then remained identical through round $7999$, covering the $462$ rounds that follow it
(Fig.~\ref{fig:trajectory}).

Because $m=\lfloor k/2\rfloor$, round $7538$ is the round at which the \emph{second} preference
cell turns over, not the first. In this trajectory the first cell turns over at round $4625$,
and realized participating declarations occur from the first round onward---at round $7536$, for
instance, all five units declared participation. Round $7538$ is therefore a first-passage event
of the frozen coarse statistic, not the onset of physical participation. We report it because it
is the event the frozen analysis tracks and because it is the point at which the two executions
could most visibly have diverged.

For context, and not as a further result: under the profile-mean estimator at this operating
point and decay structure, the canonical $96$-seed analysis records a median first-passage round of
$7537$, with $93$ of $96$ seeds reaching $m\ge1$ within the horizon and a modal $m$ of zero. The validated seed is therefore not an anomalously late-transitioning example of
its own condition. This is a statement about where the validated trajectory sits among comparable
trajectories; it is not a replication of the canonical median, and the single distributed run
carries no statistical weight for the $96$-seed distribution.

As Section~\ref{sec:integrity} states, the agreement between the two executions on the transition
round is not statistically independent evidence: given Definition~\ref{def:fid} over the horizon,
it follows. Its value is interpretive---it demonstrates that the validated interval contains the
late coarse-preference first passage rather than only a long quiet prefix in which the population never left
the basin.

\begin{figure}[t]
\centering
\includegraphics[width=\columnwidth]{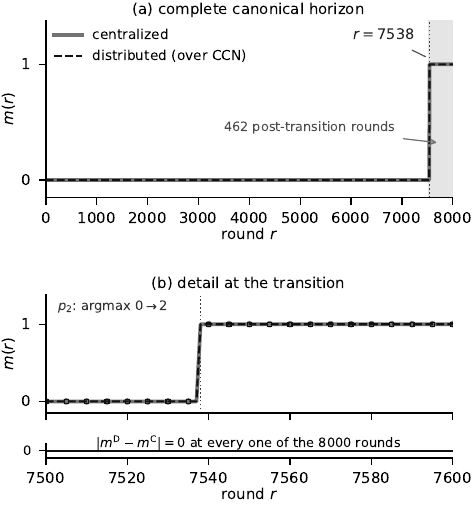}
\caption{Exact reproduction of the canonical coarse-preference trajectory by the distributed
realization, under the profile-mean estimator. (a) Coarse preference level $m(r)$ of \eqref{eq:mdef} over the
full $8000$-round canonical horizon for the centralized reference (solid) and the distributed
execution over CCN (dashed); the traces are exactly equal at every round, so the dashed trace
lies on the solid one throughout. $m$ remains $0$ until round $7538$ and equals $1$ for the remaining $462$ rounds; $m=0$ does not
mean that no unit favors, or makes, a participating declaration. (b) Detail around the transition,
with a per-round discrepancy strip that is identically zero. Agreement on the transition round is
implied by full-trajectory equality rather than independent of it; it is shown for
interpretability, to demonstrate that the validated interval contains the late coarse-preference
first passage tracked by the frozen analysis, and not as independent statistical evidence.}
\label{fig:trajectory}
\end{figure}

\subsection{The message plane}

All $40\,000$ aggregator fetches succeeded, with zero application-level retries and zero aborts;
all six forwarders were alive at completion; the run occupied $15.76$~h at $7.09$~s per round.
These statistics are reported for reproducibility and to establish that the comparison was
obtained over an active distributed execution. Consistent with Section~\ref{sec:transport}, they
are not evidence of semantic fidelity: every failure mode enumerated there is compatible with
complete delivery and no retransmission. Transport success and semantic fidelity are reported
separately here because they are separate properties.

\subsection{Fidelity under communication impairment}

A separate, shorter validation examined whether semantic fidelity survives degraded communication
when information delivery ultimately succeeds. Under six impairment conditions applied to every
link and verified in the emulated queueing discipline at runtime---up to $200$~ms delay, $40$~ms
jitter, $15\%$ loss and $0.5$~Mbit/s---the round time rose by up to $3.45\times$ and $371$
application-level timeouts occurred over a $50$-round horizon, yet all $750$ fetches eventually
succeeded with no aborts and the joint declaration, settlement, learner state and
coarse-preference trajectory were exactly equal to the reference (Table~\ref{tab:impair}).

This isolates the scope of the claim: impairment changes when a settlement arrives, not what it
contains, and the mechanism as specified in Section~\ref{sec:model} assumes each round's
settlement is available before the next update. Fidelity here is therefore a statement about the
realization under conditions in which delivery ultimately succeeds. It is not a resilience
mechanism, not a guarantee, and not a property of the substrate; behavior under permanently
unavailable settlement would require model semantics that Section~\ref{sec:model} does not
define, and is left outside this study.

\begin{table}[t]
\caption{Fidelity under impairment, $N=5$, seed $12345$. Every condition reproduces the reference
exactly; no tolerance is applied.}
\label{tab:impair}
\centering
\begin{tabular}{lrrrr}
\toprule
condition & s/round & timeouts & fetches OK & $\hat{u}$ max diff\\
\midrule
ideal ($10$ r)      & $7.40$  & $0$   & $150/150$ & $0.0$\\
bandwidth           & $7.40$  & $0$   & $150/150$ & $0.0$\\
delay + jitter      & $7.80$  & $0$   & $150/150$ & $0.0$\\
combined            & $8.40$  & $1$   & $150/150$ & $0.0$\\
loss                & $10.40$ & $10$  & $150/150$ & $0.0$\\
severe              & $24.20$ & $67$  & $150/150$ & $0.0$\\
\midrule
ideal ($50$ r)      & $7.16$  & $0$   & $750/750$ & $0.0$\\
severe ($50$ r)     & $24.68$ & $371$ & $750/750$ & $0.0$\\
\bottomrule
\end{tabular}
\end{table}

\subsection{Negative control: transport success without semantic fidelity}
\label{sec:negctl}

The result above is what a correct implementation driven by fixed randomness should produce. To
show that the comparison could have failed, we ran one preregistered negative control, separate
from the canonical run and reported here only as a property of the validation method.

At a round fixed in advance by a mechanical rule---the first round $r\ge10$ at which
participants $1$ and $2$ declare different actions, which is round $10$---the aggregator
exchanged the attribution of two declarations it had already fetched successfully, immediately
before assembling the joint profile. Nothing else was altered: no message was dropped, delayed,
retried or malformed, and the injected round's five declaration fetches all succeeded on their
first attempt. Over the $30$-round horizon the message plane recorded $450$ retrievals, $330$
publications, and \textbf{zero} retries, timeouts and aborts.

The semantic comparisons of Section~\ref{sec:levels} nevertheless separated the two executions
(Table~\ref{tab:negctl}). Four of them diverged, all first at the injected round: the
participant-versus-aggregator declaration comparison, the joint profile ($02112$ against
$01212$), the settlement vector, and the post-update estimate $\hat{u}$. Only the estimate
divergence persisted; it was still present at the end of the horizon. Three comparisons did not
separate the executions, and each for a reason worth stating. The participant-versus-centralized
declaration comparison did not, because no participant's own behavior was altered---which is
what makes the fault an attribution error rather than a behavioral one. The visit counts $n$ did
not, because each participant still updated its own (type, declaration) counter; only the value
written into it changed. And the coarse level $m$ did not, because within this early window no
$\arg\max$ turned over. One settlement entry also matched by coincidence: the abstaining
participant is paid $0.0$ under both profiles, which is a concrete instance of equal values
masking an attribution error, and is the reason settlement agreement is not treated as evidence
of correct attribution anywhere in this paper.

The control establishes that a transport-successful execution can violate these comparisons, and
that at least one of them detects this particular fault. It does not establish that the
criterion detects arbitrary implementation faults, and no claim of completeness, fault tolerance
or robustness is made from it.

\begin{table}[t]
\caption{Negative control: one preregistered attribution swap at round $10$, $N=5$, seed $12345$,
$30$ rounds. Counts are over $30$ rounds for the per-round levels and over $30\times5=150$
round--participant comparisons for the per-participant levels; the L4 rows therefore count
\emph{learner-state comparisons}, one per participant per round, and not the scalar entries counted in
Table~\ref{tab:fidelity}. The canonical run of Table~\ref{tab:fidelity} is unaffected by this
experiment.}
\label{tab:negctl}
\centering
\begin{tabular}{lr}
\toprule
\multicolumn{2}{l}{\emph{Message plane}}\\
retrievals / publications & $450$ / $330$\\
retries, timeouts, aborts & $0$, $0$, $0$\\
injected round's fetches successful & $5$ of $5$, first attempt\\
\midrule
\multicolumn{2}{l}{\emph{Semantic comparisons} (first divergence round)}\\
L1 participant vs.\ centralized & $0$ of $150$ (---)\\
L1 participant vs.\ aggregator & $\mathbf{2}$ of $150$ (r.\ $10$)\\
L2 joint profile & $\mathbf{1}$ of $30$ (r.\ $10$)\\
L3 settlement & $\mathbf{4}$ of $150$ (r.\ $10$)\\
L4 learner $\hat{u}$ & $\mathbf{80}$ of $150$ (r.\ $10$, persists)\\
L4 learner $n$ & $0$ of $150$ (---)\\
L5 coarse level $m$ & $0$ of $30$ (---)\\
\bottomrule
\end{tabular}
\end{table}

\subsection{A platform-runtime observation}

An earlier long-duration validation attempt, under the same implementation and configuration,
terminated at round $4315$ when two forwarder processes crashed; the scientific states were
exactly equal over all $4315$ completed rounds. No repair was made and none is incorporated into
the implementation reported here: the subsequent preregistered run used a byte-identical
configuration and completed all $8000$ rounds. The episode is recorded as an observation of
platform-runtime variability rather than as a scientific discrepancy, and it establishes neither
robustness of the forwarder implementation nor long-run durability as a property.

\subsection{Scope of the empirical claim}

The validation reported here covers one seed, $N=5$, the canonical operating point under the
profile-mean estimator, and one complete $8000$-round distributed trajectory, with a separate
short-horizon impairment validation at $N=5$ over six conditions. It establishes semantic
fidelity for that configuration. It does not establish equivalence for other seeds, population sizes, operating points or network
conditions.

\subsection{Limitations}
\label{sec:limits}

The scientific evidence is bounded in five ways, each established by the experiments reported
above. It is measured at $N=5$ units. Every rate is a finite-horizon rate at the canonical
$8000$-round horizon; Section~\ref{sec:feedback} shows that at $96\,000$ rounds every condition
converges, including the mechanism with no corrective transfer, so no statement here is a
statement about eventual reachability. The incumbent value $\hat{u}^{(0)}=0.20$ is an
analyst-set prior rather than a fitted quantity, and the separation is sensitive to it: the
sweep of Section~\ref{sec:window} places the lower edge of the window between $0.10$ and
$0.15$, where the margin $\delta(j,k)$ of Section~\ref{sec:where} says it should be, and at $0.10$ the linear
structure retains every unit. The separation closes entirely when counterfactual payoffs are
supplied, so the result concerns what an experience-based learner can discover under bandit
feedback and not the incentive as an object. Finally, the semantic-fidelity validation covers a
single distributed configuration and a single seed under the profile-mean construction.

No claim is made about robustness, about durability of the forwarder implementation, about any
advantage of the chosen substrate over alternatives, or about behavior under network conditions
other than those reported.

% ============================================================ VIII
\section{Conclusion}
\label{sec:conclusion}

Two corrective transfers that satisfy the same five static criteria, with the same entry
thresholds, did not produce the same participation once owners learned from their own
settlements. One reached full participation from a collapse initialization in 96 of 96 seeds and
the other in none, within the canonical horizon and with disjoint confidence intervals. The
structures agree at the two profiles the static criteria inspect and differ only in between, so
the checks that accepted both were blind to the region that decided the outcome. Programs of
this kind should therefore examine what a transfer pays at partial participation, which can be
read off the mechanism before any simulation is run.

Distributing the mechanism raises a separate question, because a unit's settlement depends on
what the others declared and is then absorbed into state that shapes its later declarations. We
compared a distributed realization against the centralized reference at each participant's
declaration, the assembled profile, the settlements and the learner states, round by round. The
two agreed at every one of the 8000 rounds. Since that is what a correct implementation should
produce, we also injected one misattributed declaration after successful delivery: the
comparison separated the executions while the message plane reported no loss, retry or timeout.
Delivered messages are not evidence that the mechanism was executed.  The evidence covers five
units, one operating point, a finite horizon and one distributed seed, and the separation between the two transfers depends on the incumbent prior
and disappears when learners are given counterfactual payoffs. Section~\ref{sec:limits} states
these bounds in full.

\appendices

\section{Proof of Theorem~\ref{th:inv}}
\label{app:inv}

No payment term appears in the dispatch program of Section~\ref{sec:model}-B, so the allocation
at every declaration profile is that of the uncorrected mechanism. At the intended profile every
unit participates, so the indicator in \eqref{eq:transfer} is one and $R_i=s_0 f(Q_{-i})$. The conservative item
is the smallest participating declaration, so the leave-one-out capability is least when all
others declare it and satisfies $Q_{-i}\ge (N-1)q_C$; by (A3) $f$ vanishes at and above
$\bar{Q}$. Whenever $(N-1)q_C \ge \bar{Q}$---here $10.0$ against $9.0$, and under the
proportional rule $\bar{Q}(N)=0.60\,N q_A$ for every $N\ge 4$ at the present $q_C,q_A$---the
transfer is identically zero at the intended profile, for every $i$ and every $s_0$, so every
payment there equals its pre-correction value. A contract-selection margin is a difference
between two \emph{participating} declarations available to $i$ with the other declarations held
fixed; the indicator in \eqref{eq:transfer} is therefore one in both, and by (A1) the remaining
factor reads $Q_{-i}$ alone, so $R_i$ adds the same constant to each and cancels. \hfill$\blacksquare$

\section{Proof of Corollary~\ref{th:cor}}
\label{app:cor}

By \eqref{eq:thresholds}, $s_0^{\mathrm{elim}}$ is the smallest constant at which some
unilateral deviation from zero participation is profitable. At zero participation a deviating
unit has $Q_{-i}=0$, so by (A2) it receives exactly $s_0$, while by (A0) the units that continue
to abstain receive nothing; its deviation gain rises by $s_0$, so for any larger constant the
profile admits a profitable deviation. Invariance of the calibration is
Theorem~\ref{th:inv}. \hfill$\blacksquare$

\bibliographystyle{IEEEtran}
\bibliography{refs}

\end{document}